\documentclass[12pt]{amsart}
\usepackage{amssymb}
\usepackage{amsmath}
\usepackage{amsfonts}
\usepackage{geometry}

\theoremstyle{plain}

\newtheorem{corollary}{Corollary}

\newtheorem{proposition}{Proposition}

\numberwithin{equation}{section}
\input{tcilatex}
\begin{document}
\title{Cartel Stability under Quality Differentiation}
\author{Iwan Bos }
\author{Marco A. Marini\medskip }
\thanks{We appreciate the comments of Maria Rosa Battaggion, Boris Ginzburg,
Joseph E. Harrington Jr., Ronald Peeters, an anonymous referee, participants
at the 2018 Oligo Workshop and the 2018 EARIE conference in Athens. All
opinions and errors are ours alone.\medskip \\
Iwan Bos, Department of Organization \& Strategy, School of Business and
Economics, Maastricht University. E-mail:\
i.bos@maastrichtuniversity.nl.\medskip \\
Marco A. Marini, Department of Social and Economic Sciences, University of
Rome La Sapienza. E-mail:\ marini@dis.uniroma1.it.}
\date{October 2018}

\begin{abstract}
This note considers cartel stability when the cartelized products are
vertically differentiated. If market shares are maintained at pre-collusive
levels, then the firm with the lowest competitive price-cost margin has the
strongest incentive to deviate from the collusive agreement. The
lowest-quality supplier has the tightest incentive constraint when the
difference in unit production costs is sufficiently small. \medskip \medskip

\noindent%
\textbf{Keywords:\ }\textit{Cartel Stability},\textit{\ Collusion, Vertical
Differentiation}.$\bigskip $\medskip

\noindent%
\end{abstract}

\maketitle

\section{Introduction}

In this note, we examine cartel stability when the cartelized products are
vertically differentiated. Goods or services are differentiated vertically
when there is consensus among consumers about how to rank them quality-wise;
comparing products A and B, all agree A to have a higher (perceived) value
than B or \textit{vice versa}. There might, however, still be a demand for
lower-quality goods when buyers face budget constraints or differ in their
willingness to pay for quality. This creates an incentive for suppliers to
compete through offering different price-quality combinations.

One implication of this price-quality dispersion is that firms that consider
colluding typically face heterogeneous incentive constraints. The fact that
firms are induced to charge different prices, for example, affects both
collusive and noncollusive profits. From a supply-side perspective, there
commonly exists a positive relationship between the quality of a good and
its production costs. This, too, impacts both sides of the constraint. It is
therefore \textit{a priori} unclear how quality differentiation impacts the
sustainability of collusion.

The scarce literature on this topic provides mixed results and, moreover,
does not consider the potential impact of cost heterogeneity.\footnote{%
An extensive and detailed overview of this literature is provided by Marini
(2018).} Assuming identical costs, H\"{a}ckner (1994) and Symeonidis (1999)
both analyze an infinitely repeated vertically differentiated duopoly game.%
\footnote{%
Apart from being analytically convenient, the identical cost assumption can
be defended on the grounds that the difference in quality may mainly come
from upfront sunk investments in which case the impact on prices would be
limited.} H\"{a}ckner (1994) studies a variation of the setting in
Gabszewicz and Thisse (1979) and Shaked and Sutton (1982) and finds that it
is the high-quality supplier who has the strongest incentive to deviate. By
contrast, Symeonidis (1999) considers a representative consumer model with
horizontal and vertical product differentiation and establishes that it is
the low-quality seller who is most eager to leave the cartel.

In the following, we analyze an $n$-firm infinitely repeated game version of
the classic vertical differentiation model of Mussa and Rosen (1978) where
production costs are assumed to be increasing in quality.\footnote{%
Like the model in H\"{a}ckner (1994), this is a model with heterogeneous
customers. Apart from the number of firms and cost heterogeneity, it differs
in terms of consumers' utility specification.} Under the assumption that
colluding firms maintain their pre-collusive market shares, we find that it
is the competitive mark-up rather than the quality of the product that
drives the incentive to deviate. Specifically, it is the supplier with the
lowest noncooperative price-cost margin who has the strongest incentive to
chisel on the cartel. Moreover, our analysis confirms the above-mentioned
conclusion by Symeonidis (1999) when the difference in unit costs is
sufficiently small.

The next section presents the model. Section 3 contains the main finding.
Section 4 concludes.

\section{Model}

There is a given set of suppliers, denoted $N=\left\{ 1,\ldots ,n\right\} $,
who repeatedly interact over an infinite, discrete time horizon. In every
period $t\in 
\mathbb{N}
$, they simultaneously make price decisions with the aim to maximize the
expected discounted sum of their profit stream. Firms face a common discount
factor $\delta \in \left( 0,1\right) $ and all prices set up until $t-1$ are
assumed public knowledge.

Each firm $i\in N$ sells a single variant of the product with quality $v_{i}$%
. We assume $\infty >v_{n}>v_{n-1}>...>v_{1}>0$ and refer to firm $n$ as the 
\textit{top firm}, firm $1$ as the \textit{bottom firm} and all others as 
\textit{intermediate firms}. Unit production costs of firm $i\in N$ are
given by the constant $c_{i}$ and we suppose these costs to be positive and
(weakly) increasing in quality, \textit{i.e.}, $c_{n}\geq c_{n-1}\geq \ldots
\geq c_{1}>0$.

Consumers have a valuation for the various product types of $\theta $, which
is uniformly distributed on $[\underline{\theta },\overline{\theta }]\subset 
\mathcal{%
\mathbb{R}
}_{++}$ with mass normalized to one. A higher $\theta $ corresponds to a
higher gross utility when consuming variant $v_{i}$. Buyers purchase no more
than one item so that someone `located' at $\theta $ obtains the following
utility

\begin{equation}
U(\theta )=\left\{ 
\begin{array}{c}
\theta v_{i}-p_{i}\text{ when buying from firm }i \\ 
0\text{ when not buying,}%
\end{array}%
\right.  \label{utility}
\end{equation}

\noindent%
where $p_{i}$ $\in \left[ 0,\overline{\theta }v_{n}\right] $ is the price
set by firm $i$.\footnote{%
Note that none of the buyers would buy at prices in excess of $\overline{%
\theta }v_{n}$.} Using (\ref{utility}),\ it can be easily verified that a
consumer at $\theta _{i}\in \lbrack \underline{\theta },\overline{\theta }]$
is indifferent between buying from, say, firm $i+1$ and firm $i$ when

\begin{equation}
\theta _{i}(p_{i},p_{i+1})=\frac{p_{i+1}-p_{i}}{v_{i+1}-v_{i}},
\label{theta}
\end{equation}%
for every $i=1,2,...,n-1$. In the ensuing analysis, we further assume that
the market is and remains \textit{covered} (\textit{i.e.}, all consumers buy
a product).\footnote{%
This is a common assumption in contributions that employ this type of
spatial setting. See, for example, Tirole (1988, pp.296-298) and Ecchia and
Lambertini (1997). We discuss some implications of this assumption at the
end of Section 3 and consider the possibility of an uncovered collusive
market in an online appendix to this paper.}

Current profit of the bottom firm ($i=1$)\ is therefore given by%
\begin{equation}
\pi _{1}\left( p_{1},p_{2}\right) =\left( p_{1}-c_{1}\right) \cdot \left(
\theta _{1}-\underline{\theta }\right) ,  \label{profit-bottom}
\end{equation}%
where $\theta _{1}=\theta _{1}(p_{1},p_{2})$ is as specified by (\ref{theta}%
). For each intermediate firm ($i=2,3,...,n-1$) profit is 
\begin{equation}
\pi _{i}\left( p_{i-1},p_{i},p_{i+1}\right) =\left( p_{i}-c_{i}\right) \cdot
\left( \theta _{i}-\theta _{i-1}\right) ,  \label{profit -intermediate}
\end{equation}%
and for the top firm ($i=n$) it is%
\begin{equation}
\pi _{n}\left( p_{n-1},p_{n}\right) =\left( p_{n}-c_{n}\right) \cdot \left( 
\overline{\theta }-\theta _{n-1}\right) .  \label{profit top}
\end{equation}

Before analyzing the infinitely repeated version of the above game, let us
first consider the one-shot case in more detail. In this setting, each firm
simultaneously picks a price to maximize its profit as specified in (\ref%
{profit-bottom})-(\ref{profit top}). Following the first-order conditions,
this yields three types of best-response functions:

\begin{equation}
\widehat{p}_{1}(p_{2})=\frac{1}{2}\left( p_{2}+c_{1}-\underline{\theta }%
(v_{2}-v_{1})\right)  \label{bottom-best}
\end{equation}%
for the bottom firm ($i=1$). For each intermediate firm ($i=2,3,...,n-1$),
the best-reply is given by

\begin{equation}
\widehat{p}_{i}(p_{i-1},p_{i+1})=\dfrac{1}{2}\dfrac{%
p_{i-1}(v_{i+1}-v_{i})+p_{i+1}(v_{i}-v_{i-1})}{(v_{i+1}-v_{i-1})}+\dfrac{1}{2%
}c_{i}.  \label{inter-best}
\end{equation}%
The best-response function of the top firm ($i=n$) is 
\begin{equation}
\widehat{p}_{n}(p_{n-1})=\frac{1}{2}\left( p_{n-1}+c_{n}+\overline{\theta }%
(v_{n}-v_{n-1})\right) .  \label{top-best}
\end{equation}%
Since the action sets are compact and convex and the above best-reply
functions are \textit{contractions}, there exists a unique static Nash
equilibrium price vector $p^{\ast }$ for any finite number of firms.%
\footnote{%
See, for instance, Friedman (1991, p.84). A sufficient condition for the
contraction property to hold is (see, for example, Vives 2000, p.47):%
\begin{equation*}
\frac{\partial ^{2}\pi _{i}}{\partial \left( p_{i}\right) ^{2}}+\tsum_{j\neq
i}\left\vert \frac{\partial ^{2}\pi _{i}}{\partial p_{i}\partial p_{j}}%
\right\vert <0,
\end{equation*}%
which, using (\ref{profit -intermediate}) for all intermediate firms $%
i=2,...,n-1$, becomes%
\begin{equation*}
\frac{v_{i-1}-v_{i+1}}{\left( v_{i+1}-v_{i}\right) \left(
v_{i}-v_{i-1}\right) }<0,
\end{equation*}%
which holds. The same applies for the top and the bottom firm.} Finally, we
impose two more conditions to ensure that the equilibrium solution is
interior (\textit{i.e.}, all firms have a positive output at $p^{\ast }$)
and that the market is indeed covered at the single-shot Nash equilibrium: 
\begin{equation}
\overline{\theta }>\theta _{n-1}^{\ast }>\theta _{n-2}^{\ast }>\ldots
>\theta _{i}^{\ast }>...>\theta _{1}^{\ast }>\underline{\theta }>\frac{%
p_{1}^{\ast }}{v_{1}}>0,  \label{H1}
\end{equation}%
where $\theta _{i}^{\ast }\equiv $ $\theta _{i}\left( p_{i}^{\ast
},p_{i+1}^{\ast }\right) $ and $p_{i}^{\ast }\geq c_{i}$, for all $i\in N$.

\section{Sustainability of Collusion}

Within the above framework, we now study the sustainability of collusion
assuming a standard grim-trigger punishment strategy. The incentive
compatibility constraint (ICC) of a firm $i\in N$ is then given by:

\begin{equation}
\Omega _{i}\equiv \pi _{i}^{c}-\left( 1-\delta \right) \cdot \pi
_{i}^{d}-\delta \cdot \pi _{i}^{\ast }\geq 0,  \label{ICC}
\end{equation}%
where $\pi _{i}^{c}=\pi _{i}\left( p_{i-1}^{c},p_{i}^{c},p_{i+1}^{c}\right) $
is its collusive payoff, $\pi _{i}^{d}=\pi _{i}\left(
p_{i-1}^{c},p_{i}^{d},p_{i+1}^{c}\right) $, with $p_{i}^{d}=\widehat{p}%
_{i}(p_{i-1}^{c},p_{i+1}^{c})$, is its deviation payoff and $\pi _{i}^{\ast
}=\pi _{i}\left( p_{i-1}^{\ast },p_{i}^{\ast },p_{i+1}^{\ast }\right) $ is
its Nash equilibrium payoff.\smallskip\ A cartel comprising the entire
industry is thus sustainable only when $\Omega _{i}\geq 0$ for all $i\in N$%
.\smallskip

In principle, this set-up allows for a plethora of sustainable collusive
contracts. In the following, we limit ourselves to what is perhaps the
simplest possible agreement. Specifically, we consider the maximization of
total cartel profits without side payments under the assumption that firms
maintain their market shares at pre-collusive levels. Such an agreement is
appealing for several reasons. First, it seems a natural focal point in the
issue of how to divide the market. Second, there have been quite a few
cartels that employed such (or similar)\ market-sharing scheme.\footnote{%
See, for example, Harrington (2006).} Third, it is arguably one of the most
subtle arrangements in that firm behavior maintains a competitive
appearance, thereby minimizing the possibility of cartel detection.

Let us now address the question of what collusive price vector such an
all-inclusive cartel would pick. As an initial observation, notice that the
fixed market share assumption implies that the price ranking should stay the
same (\textit{i.e.}, collusive prices are strictly increasing in quality).
Moreover, as market size is given, the lowest-valuation buyer should still
be willing to buy the product. This means that

\begin{equation*}
\underline{\theta }v_{1}-p_{1}^{c}\geq 0.
\end{equation*}

Next, note that the fixed market share rule in combination with the covered
market assumption implies that each `marginal consumer's location' remains
unaffected. Specifically, the consumer who was indifferent between firm $1$
and $2$ absent collusion has now the following utility when buying from firm 
$1$:%
\begin{equation*}
U(\theta _{1}^{\ast })=\left( \frac{p_{2}^{\ast }-p_{1}^{\ast }}{v_{2}-v_{1}}%
\right) v_{1}-p_{1}^{c},
\end{equation*}%
which in turn determines the collusive price for the product of firm $2$: 
\begin{equation*}
\left( \frac{p_{2}^{\ast }-p_{1}^{\ast }}{v_{2}-v_{1}}\right)
v_{1}-p_{1}^{c}=\left( \frac{p_{2}^{\ast }-p_{1}^{\ast }}{v_{2}-v_{1}}%
\right) v_{2}-p_{2}^{c}.
\end{equation*}%
Rearranging gives,%
\begin{equation*}
p_{2}^{c}=p_{2}^{\ast }+\left( p_{1}^{c}-p_{1}^{\ast }\right) .
\end{equation*}%
A higher collusive price by firm 2 would mean that the customer on the
boundary prefers firm $1$, which contradicts market shares being fixed.
Likewise, a lower price implies a decrease in sales for firm 1 and therefore
cannot occur either. The collusive prices for all other firms can be
determined in a similar fashion. In general, the collusive price of firm $%
i\in N\backslash \{1\}$ \ is equal to its Nash price plus the price increase
by the lowest-quality firm: 
\begin{equation}
p_{i}^{c}=p_{i}^{\ast }+\left( p_{1}^{c}-p_{1}^{\ast }\right) .
\label{collusive prices}
\end{equation}%
Both cartel profits and the incentive constraints are therefore effectively
a function of $p_{1}^{c}$ alone.

Since prices are strategic complements, it is clear that the cartel would
like to set $p_{1}^{c}=\underline{\theta }v_{1}$. The collusive contract
with $p_{1}^{c}=\underline{\theta }v_{1}$ might not be sustainable, however,
because one or more ICC's may be binding. The next result shows that it is
the supplier with the lowest noncollusive profit margin who has the tightest
incentive constraint.

\begin{proposition}
For any $i,j\in N$ and $j\neq i$, if $p_{i}^{\ast }-c_{i}>p_{j}^{\ast
}-c_{j} $, then $\Omega _{i}>\Omega _{j}$.
\end{proposition}

\begin{proof}
Consider the ICC of an intermediate firm $i=2,3,...,n-1$:%
\begin{equation*}
\Omega _{i}\equiv \pi _{i}^{c}-\left( 1-\delta \right) \cdot \pi
_{i}^{d}-\delta \cdot \pi _{i}^{\ast }\geq 0\Leftrightarrow \delta \geq 
\overline{\delta }_{i}\equiv \frac{\pi _{i}^{d}-\pi _{i}^{c}}{\pi
_{i}^{d}-\pi _{i}^{\ast }}.
\end{equation*}

To evaluate the critical discount factor $\overline{\delta }_{i}$ of every
intermediate firm, let us focus on $\pi _{i}^{d}$, $\pi _{i}^{\ast }$ and $%
\pi _{i}^{c}$ in turn. Every firm's $i=2,3,...,n-1$ deviating profit is
given by $\pi _{i}^{d}=\left( p_{i}^{d}-c_{i}\right) \left( \theta
_{i}^{d}-\theta _{i-1}^{d}\right) $, which, using best-replies (\ref%
{inter-best}), yields 
\begin{equation*}
p_{i}^{d}-c_{i}=\frac{%
p_{i-1}^{c}(v_{i+1}-v_{i})+p_{i+1}^{c}(v_{i}-v_{i-1})-(v_{i+1}-v_{i-1})c_{i}%
}{2(v_{i+1}-v_{i-1})}
\end{equation*}%
or%
\begin{equation}
2(v_{i+1}-v_{i-1})\left( p_{i}^{d}-c_{i}\right)
=p_{i-1}^{c}(v_{i+1}-v_{i})+p_{i+1}^{c}(v_{i}-v_{i-1})-(v_{i+1}-v_{i-1})c_{i}.\smallskip
\label{1}
\end{equation}%
$\smallskip $

Moreover,%
\begin{equation*}
\theta _{i}^{d}-\theta _{i-1}^{d}=\frac{\left( v_{i}-v_{i-1}\right)
p_{i+1}^{c}+\left( v_{i+1}-v_{i}\right) p_{i-1}^{c}-\left(
v_{i+1}-v_{i-1}\right) c_{i}}{2\left( v_{i+1}-v_{i}\right) \left(
v_{i}-v_{i-1}\right) },\smallskip
\end{equation*}%
from which\smallskip $\smallskip $%
\begin{equation}
2\left( v_{i+1}-v_{i}\right) \left( v_{i}-v_{i-1}\right) \left( \theta
_{i}^{d}-\theta _{i-1}^{d}\right) =\left( v_{i}-v_{i-1}\right)
p_{i+1}^{c}+\left( v_{i+1}-v_{i}\right) p_{i-1}^{c}-\left(
v_{i+1}-v_{i-1}\right) c_{i}.  \label{2}
\end{equation}%
$\smallskip $

Combining (\ref{1})-(\ref{2}) yields\smallskip $\smallskip $%
\begin{equation*}
2(v_{i+1}-v_{i-1})\left( p_{i}^{d}-c_{i}\right) =2\left(
v_{i+1}-v_{i}\right) \left( v_{i}-v_{i-1}\right) \left( \theta
_{i}^{d}-\theta _{i-1}^{d}\right) \smallskip
\end{equation*}%
and therefore%
\begin{equation*}
\left( \theta _{i}^{d}-\theta _{i-1}^{d}\right) =\frac{(v_{i+1}-v_{i-1})}{%
\left( v_{i+1}-v_{i}\right) \left( v_{i}-v_{i-1}\right) }\left(
p_{i}^{d}-c_{i}\right) .\smallskip
\end{equation*}

Hence, deviating profits can be written as\smallskip 
\begin{equation}
\pi _{i}^{d}=\left( p_{i}^{d}-c_{i}\right) ^{2}\left( \frac{(v_{i+1}-v_{i-1})%
}{\left( v_{i+1}-v_{i}\right) \left( v_{i}-v_{i-1}\right) }\right)
.\smallskip  \label{pi d}
\end{equation}

In a similar vein, intermediate firms' Nash profit can be written
as\smallskip 
\begin{equation}
\pi _{i}^{\ast }=\left( p_{i}^{\ast }-c_{i}\right) ^{2}\left( \frac{%
(v_{i+1}-v_{i-1})}{\left( v_{i+1}-v_{i}\right) \left( v_{i}-v_{i-1}\right) }%
\right) .\smallskip  \label{pi star}
\end{equation}

Following the covered market assumption, we know that 
\begin{equation*}
\theta _{i}^{c}-\theta _{i-1}^{c}=\theta _{i}^{\ast }-\theta _{i-1}^{\ast }
\end{equation*}%
and 
\begin{equation*}
p_{i}^{c}=p_{i}^{\ast }+p_{1}^{c}-p_{1}^{\ast }.
\end{equation*}%
$\smallskip $Hence, every intermediate firm's collusive profit can be
written as\smallskip\ 
\begin{equation}
\pi _{i}^{c}=\left( p_{i}^{c}-c_{i}\right) \left( \theta _{i}^{c}-\theta
_{i-1}^{c}\right) =\left( p_{i}^{\ast }+p_{1}^{c}-p_{1}^{\ast }-c_{i}\right)
\left( \frac{(v_{i+1}-v_{i-1})}{\left( v_{i+1}-v_{i}\right) \left(
v_{i}-v_{i-1}\right) }\left( p_{i}^{\ast }-c_{i}\right) \right) .
\label{pi c}
\end{equation}%
$\smallskip $

Note further that$\smallskip $\ 
\begin{eqnarray*}
p_{i}^{d} &=&\frac{\left( v_{i}-v_{i-1}\right) p_{i+1}^{c}+\left(
v_{i+1}-v_{i}\right) p_{i-1}^{c}+\left( v_{i+1}-v_{i-1}\right) c_{i}}{%
2\left( v_{i+1}-v_{i-1}\right) } \\
&=&\frac{\left( v_{i}-v_{i-1}\right) \left( p_{i+1}^{\ast
}+p_{1}^{c}-p_{1}^{\ast }\right) +\left( v_{i+1}-v_{i}\right) \left(
p_{i-1}^{\ast }+p_{1}^{c}-p_{1}^{\ast }\right) +\left(
v_{i+1}-v_{i-1}\right) c_{i}}{2\left( v_{i+1}-v_{i-1}\right) } \\
&=&p_{i}^{\ast }+\frac{1}{2}\left( p_{1}^{c}-p_{1}^{\ast }\right) .
\end{eqnarray*}

Combining (\ref{pi d})-(\ref{pi c}) yields$\smallskip \smallskip $%
\begin{equation}
\pi _{i}^{d}-\pi _{i}^{c}=\left( \frac{(v_{i+1}-v_{i-1})}{\left(
v_{i+1}-v_{i}\right) \cdot \left( v_{i}-v_{i-1}\right) }\right) \left[ \frac{%
1}{4}\left( p_{1}^{c}-p_{1}^{\ast }\right) ^{2}\right] ,
\label{pi d minus pi c}
\end{equation}%
$\smallskip \smallskip $and$\smallskip $%
\begin{equation}
\pi _{i}^{d}-\pi _{i}^{\ast }=\left( \frac{(v_{i+1}-v_{i-1})}{\left(
v_{i+1}-v_{i}\right) \cdot \left( v_{i}-v_{i-1}\right) }\right) \left(
p_{1}^{c}-p_{1}^{\ast }\right) \left[ \frac{1}{4}\left(
p_{1}^{c}-p_{1}^{\ast }\right) +\left( p_{i}^{\ast }-c_{i}\right) \right] .
\label{pi d minus pi star}
\end{equation}%
$\smallskip \smallskip $Thus, by (\ref{pi d minus pi c}) and (\ref{pi d
minus pi star}), the critical discount factor of every intermediate firm $%
i=2,3,...,n-1$ is given by\smallskip $\smallskip $ 
\begin{equation*}
\overline{\delta }_{i}\equiv \frac{\pi _{i}^{d}-\pi _{i}^{c}}{\pi
_{i}^{d}-\pi _{i}^{\ast }}=\frac{\frac{1}{4}\left( p_{1}^{c}-p_{1}^{\ast
}\right) }{\frac{1}{4}\left( p_{1}^{c}-p_{1}^{\ast }\right) +\left(
p_{i}^{\ast }-c_{i}\right) }\smallskip ,
\end{equation*}

which is decreasing in the noncollusive price-cost margin.$\smallskip $

Turning to the top-quality firm, by following the above steps, it can be
verified that$\smallskip $%
\begin{eqnarray*}
\pi _{n}^{d} &=&\frac{\left( p_{n}^{d}-c_{n}\right) ^{2}}{\left(
v_{n}-v_{n-1}\right) },\text{ }\pi _{n}^{\ast }=\frac{\left( p_{n}^{\ast
}-c_{n}\right) ^{2}}{\left( v_{n}-v_{n-1}\right) }\text{ and }\smallskip \\
&& \\
\pi _{n}^{c} &=&\left( p_{1}^{c}-p_{1}^{\ast }+p_{n}^{\ast }-c_{n}\right)
\left( \frac{\overline{\theta }\left( v_{n}-v_{n-1}\right) -\left(
p_{n}^{\ast }-p_{n-1}^{\ast }\right) }{v_{n}-v_{n-1}}\right) ,\smallskip 
\text{ }
\end{eqnarray*}

yielding$\smallskip $%
\begin{equation*}
\overline{\delta }_{n}\equiv \frac{\pi _{n}^{d}-\pi _{n}^{c}}{\pi
_{n}^{d}-\pi _{n}^{\ast }}=\frac{\frac{1}{4}\left( p_{1}^{c}-p_{1}^{\ast
}\right) }{\frac{1}{4}\left( p_{1}^{c}-p_{1}^{\ast }\right) +\left(
p_{n}^{\ast }-c_{n}\right) }.
\end{equation*}%
$\smallskip $

Hence, it is the supplier with the smaller noncooperative price-cost margin
(between firm $n$ and firm $n-1$) who has the tighter incentive constraint.

Finally, in a similar fashion, it can be shown that%
\begin{eqnarray*}
\pi _{1}^{d} &=&\frac{\left( p_{1}^{d}-c_{1}\right) ^{2}}{\left(
v_{2}-v_{1}\right) },\text{ }\pi _{1}^{\ast }=\frac{\left( p_{1}^{\ast
}-c_{1}\right) ^{2}}{\left( v_{2}-v_{1}\right) }\text{ and }\smallskip \\
&& \\
\pi _{1}^{c} &=&\left( p_{1}^{c}-c_{1}\right) \left( \frac{%
p_{2}^{c}-p_{1}^{c}-\left( v_{2}-v_{1}\right) \underline{\theta }}{%
v_{2}-v_{1}}\right) =\frac{\left( p_{1}^{c}-c_{1}\right) \left( p_{1}^{\ast
}-c_{1}\right) }{\left( v_{2}-v_{1}\right) },\smallskip \text{ }
\end{eqnarray*}

yielding$\smallskip $%
\begin{equation*}
\overline{\delta }_{1}\equiv \frac{\pi _{1}^{d}-\pi _{1}^{c}}{\pi
_{1}^{d}-\pi _{1}^{\ast }}=\frac{\frac{1}{4}\left( p_{1}^{c}-p_{1}^{\ast
}\right) }{\frac{1}{4}\left( p_{1}^{c}-p_{1}^{\ast }\right) +\left(
p_{1}^{\ast }-c_{1}\right) }.
\end{equation*}

Hence, $p_{2}^{\ast }-c_{2}>p_{1}^{\ast }-c_{1}$ implies $\Omega _{2}>\Omega
_{1}$, whereas $p_{1}^{\ast }-c_{1}>p_{2}^{\ast }-c_{2}$ implies $\Omega
_{1}>\Omega _{2}$. We thus conclude that, if $p_{i}^{\ast
}-c_{i}>p_{j}^{\ast }-c_{j}$, then $\Omega _{i}>\Omega _{j}$ for all $i,j\in
N,$ $j\neq i$.
\end{proof}

The incentive to deviate from the collusive agreement is determined by the
short-term gain of defection $\left( \pi _{i}^{d}-\pi _{i}^{c}\right) $ and
the severity of the resulting punishment ($\pi _{i}^{d}-\pi _{i}^{\ast }=\pi
_{i}^{d}-\pi _{i}^{c}+\pi _{i}^{c}-\pi _{i}^{\ast }$). The proof of
Proposition 1 reveals that the `extra profit effect' is the same across
firms and that the heterogeneity in incentive constraints is exclusively
driven by differences in the punishment impact $\left( \pi _{i}^{c}-\pi
_{i}^{\ast }\right) $. Specifically, there is a positive relation between a
firm's market share and its noncollusive price-cost margin. Since market
shares are fixed at pre-collusive levels, members with a higher competitive
mark-up are hit relatively harder by a cartel breakdown and this creates a
stronger incentive to abide by the agreement.

The next result follows immediately.\ In stating this result, let $\triangle
c_{ij}=c_{i}-c_{j}$ for any firm $i,j\in N$ and $j\neq i$.

\begin{corollary}
For any firm $i,j\in N$ and $j\neq i$, $\exists \mu $ $\in $ $\mathcal{%
\mathbb{R}
}_{++}$ such that if $\triangle c_{ij}<\mu $ and $v_{i}>v_{j}$\textit{, then 
}$\Omega _{i}>\Omega _{j}$.
\end{corollary}

Hence, if the difference in unit production costs is sufficiently small,
then it is the lowest-quality supplier who has the tightest incentive
constraint. Our analysis thus confirms the above-mentioned conclusion by
Symeonidis (1999) in case quality heterogeneity is primarily driven by
(sunk) fixed costs rather than variable costs.\footnote{%
This result can also be shown to hold in an $n$-firm variation of the model
in H\"{a}ckner (1994). The analysis is available upon request.}

Let us conclude this section with a remark on the covered market assumption.
Note that since each ICC\ is strictly concave in the own cartel price, a
sufficient condition for the cartel to keep market size constant is that at $%
p_{1}^{c}=\underline{\theta }v_{1}$, $\Omega _{i}\leq 0$ and $\frac{\partial
\Omega _{i}}{\partial p_{i}^{c}}<0$ for all $i\in N$. Endogenizing the size
of the market would therefore not affect the above findings when the
discount factor is sufficiently low. If its members are patient enough,
however, then the cartel would like to uncover the market. This case is
generally far less tractable analytically, but we show in an online appendix
that this paper's results also hold when the number of low-valuation buyers
leaving the market is sufficiently small.

\section{Conclusion}

Many markets are characterized by some degree of quality differentiation
with corresponding firm heterogeneity in cost and demand. One implication of
such differences is that colluding firms typically face non-identical
incentive constraints. Existing literature on this topic focuses on demand
differences, while ignoring the potential impact of cost heterogeneity. In
this note, we considered how cartel stability is affected when unit costs
are increasing in product quality. Under the assumption that colluding firms
maintain their pre-collusive market shares, we found that the incentive to
deviate from the collusive agreement is monotonic in the noncollusive
price-cost margin. Specifically, the supplier with the lowest competitive
mark-up is \textit{ceteris paribus} most inclined to leave the cartel.
Moreover, it is the lowest-quality seller who has the tightest incentive
constraint when differences in unit costs are sufficiently small.

\section{Appendix 1:\ Uncovered Market Case}

In the note we show that the incentive to deviate from a cartel agreement is
monotonic in the noncollusive price-cost margin and that the firm with the
lowest profit margin has the tightest incentive constraint (Proposition 1).
This result is derived under the assumption of a fixed market size. In this
appendix, we consider the possibility that the cartel `uncovers the market'
by setting collusive prices at a level where some of the lowest-valuation
customers prefer to no longer buy the product. In the following, we show
that the result of Proposition 1 also holds when the number of buyers
leaving the market is sufficiently small.

To begin, consider the incentive compatibility constraint (ICC)\ of an
intermediate firm $i\in N$:\smallskip

\begin{equation*}
\Omega _{i}\equiv \left( p_{i}^{c}-c_{i}\right) \left( \theta
_{i}^{c}-\theta _{i-1}^{c}\right) -\left( 1-\delta \right) \left(
p_{i}^{d}-c_{i}\right) \left( \theta _{i}^{d}-\theta _{i-1}^{d}\right)
-\delta \left( p_{i}^{\ast }-c_{i}\right) \left( \theta _{i}^{\ast }-\theta
_{i-1}^{\ast }\right) \geq 0.
\end{equation*}%
As the cartel uses a fixed market share rule, it holds that:

\begin{equation*}
\frac{\theta _{i}^{c}-\theta _{i-1}^{c}}{\overline{\theta }-\frac{p_{1}^{c}}{%
v_{1}}}=\frac{\theta _{i}^{\ast }-\theta _{i-1}^{\ast }}{\overline{\theta }-%
\underline{\theta }}.
\end{equation*}%
The ICC\ can thus be written as:

\begin{equation*}
\Omega _{i}\equiv \left( p_{i}^{c}-c_{i}\right) \left( \theta _{i}^{\ast
}-\theta _{i-1}^{\ast }\right) \frac{\left( \overline{\theta }-\frac{%
p_{1}^{c}}{v_{1}}\right) }{\left( \overline{\theta }-\underline{\theta }%
\right) }-\left( 1-\delta \right) \left( p_{i}^{d}-c_{i}\right) \left(
\theta _{i}^{d}-\theta _{i-1}^{d}\right) -\delta \left( p_{i}^{\ast
}-c_{i}\right) \left( \theta _{i}^{\ast }-\theta _{i-1}^{\ast }\right) \geq
0,
\end{equation*}%
which is equivalent to%
\begin{equation*}
\delta \geq \overline{\delta }_{i}\equiv \frac{\pi _{i}^{d}-\pi _{i}^{c}}{%
\pi _{i}^{d}-\pi _{i}^{\ast }}=\frac{\left( p_{i}^{d}-c_{i}\right) \left(
\theta _{i}^{d}-\theta _{i-1}^{d}\right) -\left( p_{i}^{c}-c_{i}\right)
\left( \theta _{i}^{\ast }-\theta _{i-1}^{\ast }\right) \frac{\left( 
\overline{\theta }-\frac{p_{1}^{c}}{v_{1}}\right) }{\left( \overline{\theta }%
-\underline{\theta }\right) }}{\left( p_{i}^{d}-c_{i}\right) \left( \theta
_{i}^{d}-\theta _{i-1}^{d}\right) -\left( p_{i}^{\ast }-c_{i}\right) \left(
\theta _{i}^{\ast }-\theta _{i-1}^{\ast }\right) }.
\end{equation*}

Let us now specify $\pi _{i}^{d}$, $\pi _{i}^{\ast }$ and $\pi _{i}^{c}$.
Following the proof of Proposition 1, deviating and Nash profits are
respectively given by%
\begin{equation*}
\pi _{i}^{d}=\left( p_{i}^{d}-c_{i}\right) ^{2}\left( \frac{(v_{i+1}-v_{i-1})%
}{\left( v_{i+1}-v_{i}\right) \left( v_{i}-v_{i-1}\right) }\right) ,
\end{equation*}%
and%
\begin{equation*}
\pi _{i}^{\ast }=\left( p_{i}^{\ast }-c_{i}\right) ^{2}\left( \frac{%
(v_{i+1}-v_{i-1})}{\left( v_{i+1}-v_{i}\right) \left( v_{i}-v_{i-1}\right) }%
\right) .
\end{equation*}%
As to collusive profits, note that uncovering the market in combination with
the fixed market share rule implies that the lowest quality firm has the
smallest price increase and that the price increase is rising in quality. In
other words, cartel prices must be chosen such that each marginal consumer's
location `shifts upwards' in order to maintain market shares at
pre-collusive levels. Rather than adding $p_{1}^{c}-p_{1}^{\ast }$ to its
Nash price $p_{i}^{\ast }$ (as in the covered market case), intermediate
firm $i$ should therefore raise its price by more. Let this additional
amount be indicated by $x_{i}>0$ so that its collusive price is given by $%
p_{i}^{c}=p_{i}^{\ast }+p_{1}^{c}-p_{1}^{\ast }+x_{i}$. Collusive profits
are then%
\begin{equation*}
\pi _{i}^{c}=\left( p_{i}^{c}-c_{i}\right) \left( \theta _{i}^{c}-\theta
_{i-1}^{c}\right) =\left( p_{i}^{\ast }+p_{1}^{c}-p_{1}^{\ast
}+x_{i}-c_{i}\right) \left( \frac{(v_{i+1}-v_{i-1})\left( p_{i}^{\ast
}-c_{i}\right) }{\left( v_{i+1}-v_{i}\right) \left( v_{i}-v_{i-1}\right) }%
\right) \frac{\left( \overline{\theta }-\frac{p_{1}^{c}}{v_{1}}\right) }{%
\left( \overline{\theta }-\underline{\theta }\right) }.
\end{equation*}%
Focusing on the numerator of the critical discount factor first, we have

\begin{eqnarray*}
\pi _{i}^{d}-\pi _{i}^{c} &=&\frac{\left( p_{i}^{d}-c_{i}\right)
^{2}(v_{i+1}-v_{i-1})}{\left( v_{i+1}-v_{i}\right) \left(
v_{i}-v_{i-1}\right) }-\left( p_{i}^{\ast }+p_{1}^{c}-p_{1}^{\ast
}+x_{i}-c_{i}\right) \left( \frac{\left( p_{i}^{\ast }-c_{i}\right)
(v_{i+1}-v_{i-1})}{\left( v_{i+1}-v_{i}\right) \left( v_{i}-v_{i-1}\right) }%
\right) \frac{\left( \overline{\theta }-\frac{p_{1}^{c}}{v_{1}}\right) }{%
\left( \overline{\theta }-\underline{\theta }\right) } \\
&& \\
&=&\left( \frac{(v_{i+1}-v_{i-1})}{\left( v_{i+1}-v_{i}\right) \left(
v_{i}-v_{i-1}\right) }\right) \left[ \left( p_{i}^{d}-c_{i}\right) ^{2}-%
\frac{\left( p_{i}^{\ast }+p_{1}^{c}-p_{1}^{\ast }+x_{i}-c_{i}\right) \left(
p_{i}^{\ast }-c_{i}\right) \left( \overline{\theta }-\frac{p_{1}^{c}}{v_{1}}%
\right) }{\left( \overline{\theta }-\underline{\theta }\right) }\right] .
\end{eqnarray*}%
In this case, the deviating price is given by%
\begin{eqnarray*}
p_{i}^{d} &=&\frac{\left( v_{i}-v_{i-1}\right) p_{i+1}^{c}+\left(
v_{i+1}-v_{i}\right) p_{i-1}^{c}+\left( v_{i+1}-v_{i-1}\right) c_{i}}{%
2\left( v_{i+1}-v_{i-1}\right) } \\
&=&\frac{\left( v_{i}-v_{i-1}\right) \left( p_{i+1}^{\ast
}+p_{1}^{c}-p_{1}^{\ast }+x_{i+1}\right) +\left( v_{i+1}-v_{i}\right) \left(
p_{i-1}^{\ast }+p_{1}^{c}-p_{1}^{\ast }+x_{i-1}\right) +\left(
v_{i+1}-v_{i-1}\right) c_{i}}{2\left( v_{i+1}-v_{i-1}\right) } \\
&=&p_{i}^{\ast }+\frac{1}{2}\left( p_{1}^{c}-p_{1}^{\ast }\right) +\frac{%
\left( v_{i}-v_{i-1}\right) x_{i+1}+\left( v_{i+1}-v_{i}\right) x_{i-1}}{%
2\left( v_{i+1}-v_{i-1}\right) }.
\end{eqnarray*}%
To facilitate the presentation of the analysis, let us denote 
\begin{equation*}
y_{i}=\frac{\left( v_{i}-v_{i-1}\right) x_{i+1}+\left( v_{i+1}-v_{i}\right)
x_{i-1}}{2\left( v_{i+1}-v_{i-1}\right) }\text{ \ and }s=\frac{\overline{%
\theta }-\frac{p_{1}^{c}}{v_{1}}}{\overline{\theta }-\underline{\theta }}.
\end{equation*}
Substituting in the above equation gives%
\begin{eqnarray*}
\pi _{i}^{d}-\pi _{i}^{c} &=&\left( \frac{(v_{i+1}-v_{i-1})}{\left(
v_{i+1}-v_{i}\right) \left( v_{i}-v_{i-1}\right) }\right) \left[ 
\begin{array}{c}
\left( p_{i}^{\ast }+\frac{1}{2}\left( p_{1}^{c}-p_{1}^{\ast }\right)
+y_{i}-c_{i}\right) ^{2} \\ 
-\left( p_{i}^{\ast }+p_{1}^{c}-p_{1}^{\ast }+x_{i}-c_{i}\right) \left(
p_{i}^{\ast }-c_{i}\right) s%
\end{array}%
\right] \\
&& \\
&=&\left( \frac{(v_{i+1}-v_{i-1})}{\left( v_{i+1}-v_{i}\right) \left(
v_{i}-v_{i-1}\right) }\right) \left[ 
\begin{array}{c}
\frac{1}{4}\left( p_{1}^{c}-p_{1}^{\ast }\right) ^{2}+\left( 1-s\right)
\left( p_{i}^{\ast }-c_{i}\right) \left[ \left( p_{i}^{\ast }-c_{i}\right)
+\left( p_{1}^{c}-p_{1}^{\ast }\right) +x_{i}\right] \\ 
+\left( p_{i}^{\ast }-c_{i}\right) \left( 2y_{i}-x_{i}\right) +y_{i}\left(
\left( p_{1}^{c}-p_{1}^{\ast }\right) +y_{i}\right)%
\end{array}%
\right] \\
&&
\end{eqnarray*}%
or 
\begin{equation}
\pi _{i}^{d}-\pi _{i}^{c}=\left( \frac{(v_{i+1}-v_{i-1})}{\left(
v_{i+1}-v_{i}\right) \left( v_{i}-v_{i-1}\right) }\right) \left[ \frac{1}{4}%
\left( p_{1}^{c}-p_{1}^{\ast }\right) ^{2}+z\right] ,
\label{pi d minus pi c}
\end{equation}%
where 
\begin{equation*}
z\equiv \left( 1-s\right) \left( p_{i}^{\ast }-c_{i}\right) \left[ \left(
p_{i}^{\ast }-c_{i}\right) +\left( p_{1}^{c}-p_{1}^{\ast }\right) +x_{i}%
\right] +\left( p_{i}^{\ast }-c_{i}\right) \left( 2y_{i}-x_{i}\right)
+y_{i}\left( \left( p_{1}^{c}-p_{1}^{\ast }\right) +y_{i}\right) .
\end{equation*}%
Note that with a covered market it holds that $s=1$, $x_{i}=0$ and $y_{i}=0$%
, in which case $z=0$ and the numerator reduces to the corresponding value
in the proof of Proposition 1. Turning to the denominator of the critical
discount factor, we have%
\begin{eqnarray*}
\pi _{i}^{d}-\pi _{i}^{\ast } &=&\left( p_{i}^{d}-c_{i}\right) ^{2}\frac{%
(v_{i+1}-v_{i-1})}{\left( v_{i+1}-v_{i}\right) \left( v_{i}-v_{i-1}\right) }%
-\left( p_{i}^{\ast }-c_{i}\right) ^{2}\frac{(v_{i+1}-v_{i-1})}{\left(
v_{i+1}-v_{i}\right) \left( v_{i}-v_{i-1}\right) } \\
&& \\
&=&\frac{(v_{i+1}-v_{i-1})\left[ \left( p_{i}^{d}-c_{i}\right) ^{2}-\left(
p_{i}^{\ast }-c_{i}\right) ^{2}\right] }{\left( v_{i+1}-v_{i}\right) \left(
v_{i}-v_{i-1}\right) } \\
&& \\
&=&\frac{(v_{i+1}-v_{i-1})\left[ \left( p_{i}^{\ast }+\frac{1}{2}\left(
p_{1}^{c}-p_{1}^{\ast }\right) +y_{i}-c_{i}\right) ^{2}-\left( p_{i}^{\ast
}-c_{i}\right) ^{2}\right] }{\left( v_{i+1}-v_{i}\right) \left(
v_{i}-v_{i-1}\right) } \\
&& \\
&=&\frac{(v_{i+1}-v_{i-1})\left[ \frac{1}{4}\left( p_{1}^{c}-p_{1}^{\ast
}\right) ^{2}+\left( p_{i}^{\ast }-c_{i}\right) \left( p_{1}^{c}-p_{1}^{\ast
}+2y_{i}\right) +y_{i}\left( \left( p_{1}^{c}-p_{1}^{\ast }\right)
+y_{i}\right) \right] }{\left( v_{i+1}-v_{i}\right) \left(
v_{i}-v_{i-1}\right) }
\end{eqnarray*}%
or 
\begin{equation}
\pi _{i}^{d}-\pi _{i}^{\ast }=\frac{(v_{i+1}-v_{i-1})\left[ \frac{1}{4}%
\left( p_{1}^{c}-p_{1}^{\ast }\right) ^{2}+r\right] }{\left(
v_{i+1}-v_{i}\right) \left( v_{i}-v_{i-1}\right) },
\label{pi d minus pi star}
\end{equation}%
where%
\begin{equation*}
r\equiv \left( p_{i}^{\ast }-c_{i}\right) \left( p_{1}^{c}-p_{1}^{\ast
}+2y_{i}\right) +y_{i}\left( \left( p_{1}^{c}-p_{1}^{\ast }\right)
+y_{i}\right) .
\end{equation*}%
\smallskip Note that in case of a covered market $y_{i}=0$ $\ $so that $r=$ $%
\left( p_{1}^{c}-p_{1}^{\ast }\right) \left( p_{i}^{\ast }-c_{i}\right) $ as
in the proof of Proposition 1. Combining both terms (\ref{pi d minus pi c})
and (\ref{pi d minus pi star}) gives the critical discount factor:\smallskip 
\begin{equation*}
\delta \geq \overline{\delta }_{i}\equiv \frac{%
\begin{array}{c}
\frac{1}{4}\left( p_{1}^{c}-p_{1}^{\ast }\right) ^{2}+\left( 1-s\right)
\left( p_{i}^{\ast }-c_{i}\right) \left[ \left( p_{i}^{\ast }-c_{i}\right)
+\left( p_{1}^{c}-p_{1}^{\ast }\right) +x_{i}\right] \\ 
+\left( p_{i}^{\ast }-c_{i}\right) \left( 2y_{i}-x_{i}\right) +y_{i}\left(
\left( p_{1}^{c}-p_{1}^{\ast }\right) +y_{i}\right)%
\end{array}%
}{\frac{1}{4}\left( p_{1}^{c}-p_{1}^{\ast }\right) ^{2}+\left( p_{i}^{\ast
}-c_{i}\right) \left( p_{1}^{c}-p_{1}^{\ast }+2y_{i}\right) +y_{i}\cdot
\left( \left( p_{1}^{c}-p_{1}^{\ast }\right) +y_{i}\right) }.
\end{equation*}%
As a final step, let us evaluate this critical discount factor with respect
to the price-cost margin $\left( p_{i}^{\ast }-c_{i}\right) $. Taking the
first-derivative with respect to the noncollusive profit margin of firm $i$
yields:

\begin{eqnarray*}
\frac{\partial \overline{\delta }_{i}}{\partial \left( \left( p_{i}^{\ast
}-c_{i}\right) \right) } &=&\frac{\left\{ 
\begin{array}{c}
\frac{1}{4}\left( p_{1}^{c}-p_{1}^{\ast }\right) ^{2}+\left( p_{i}^{\ast
}-c_{i}\right) \left( p_{1}^{c}-p_{1}^{\ast }+2y_{i}\right)  \\ 
+y_{i}\left( \left( p_{1}^{c}-p_{1}^{\ast }\right) +y_{i}\right) 
\end{array}%
\right\} \left[ 
\begin{array}{c}
\left( 1-s\right) \left[ 2\left( p_{i}^{\ast }-c_{i}\right) +\left(
p_{1}^{c}-p_{1}^{\ast }\right) +x_{i}\right]  \\ 
+\left( 2y_{i}-x_{i}\right) 
\end{array}%
\right] }{\left( \frac{1}{4}\left( p_{1}^{c}-p_{1}^{\ast }\right)
^{2}+\left( p_{i}^{\ast }-c_{i}\right) \left( p_{1}^{c}-p_{1}^{\ast
}+2y_{i}\right) +y_{i}\left( \left( p_{1}^{c}-p_{1}^{\ast }\right)
+y_{i}\right) \right) ^{2}} \\
&&-\frac{\left\{ 
\begin{array}{c}
\frac{1}{4}\left( p_{1}^{c}-p_{1}^{\ast }\right) ^{2}+\left( 1-s\right)
\left( p_{i}^{\ast }-c_{i}\right) \left[ \left( p_{i}^{\ast }-c_{i}\right)
+\left( p_{1}^{c}-p_{1}^{\ast }\right) +x_{i}\right]  \\ 
+\left( p_{i}^{\ast }-c_{i}\right) \left( 2y_{i}-x_{i}\right) +y_{i}\left(
\left( p_{1}^{c}-p_{1}^{\ast }\right) +y_{i}\right) 
\end{array}%
\right\} \left( p_{1}^{c}-p_{1}^{\ast }+2y_{i}\right) }{\left( \frac{1}{4}%
\left( p_{1}^{c}-p_{1}^{\ast }\right) ^{2}+\left( p_{i}^{\ast }-c_{i}\right)
\left( p_{1}^{c}-p_{1}^{\ast }+2y_{i}\right) +y_{i}\left( \left(
p_{1}^{c}-p_{1}^{\ast }\right) +y_{i}\right) \right) ^{2}}, \\
&&
\end{eqnarray*}%
which is negative when%
\begin{eqnarray*}
&&\left( 1-s\right) \left( p_{i}^{\ast }-c_{i}\right) \left\{ 2y_{i}\left(
\left( p_{1}^{c}-p_{1}^{\ast }\right) +y_{i}\right) +\left( p_{i}^{\ast
}-c_{i}\right) \left( p_{1}^{c}-p_{1}^{\ast }+2y_{i}\right) +\frac{1}{2}%
\left( p_{1}^{c}-p_{1}^{\ast }\right) ^{2}\right\}  \\
&& \\
&&-s\left\{ \left( p_{1}^{c}-p_{1}^{\ast }\right) +x_{i}\right\} \left[
y_{i}\left( \left( p_{1}^{c}-p_{1}^{\ast }\right) +y_{i}\right) +\frac{1}{4}%
\left( p_{1}^{c}-p_{1}^{\ast }\right) ^{2}\right] <0.
\end{eqnarray*}%
Note that this condition holds for $s\rightarrow 1$. In a similar fashion,
this can be shown to be true for the bottom and top quality firm cases.
Thus, we conclude that the result of Proposition 1 also holds when the
cartel uncovers the market and the fraction of buyers no longer buying the
product is sufficiently small.

\section{Appendix 2:\ Example with Nonuniform Distribution}

Let us present a simple two-firm example with a non-uniform distribution of $%
\theta $. The results of the paper are also shown to hold in this case.
Though of course far from being a proof, this suggests that our findings may
apply for a wider class of customer distributions.

Consider the model of the paper, but with two firms and a simple `two-step
uniform' distribution function where a subset of consumers (of mass $s$)\ is
uniformly distributed over a given interval $\left[ \underline{\theta },%
\widetilde{\theta }\right] $, with $\widetilde{\theta }\in $ $\left( 
\underline{\theta },\overline{\theta }\right) $ and another subset (of mass $%
\left( 1-s\right) \neq s$) is uniformly distributed over $\left( \widetilde{%
\theta },\overline{\theta }\right] $ (higher willingness to pay consumers).
The density function takes the form: 
\begin{equation*}
f(\theta )=\left\{ 
\begin{array}{c}
\frac{s}{(\widetilde{\theta }-\underline{\theta })}\text{ for }\theta \in %
\left[ \underline{\theta },\widetilde{\theta }\right] \text{ and} \\ 
\frac{1-s}{(\overline{\theta }-\widetilde{\theta })}\text{ for }\theta \in
\left( \widetilde{\theta },\overline{\theta }\right] .%
\end{array}%
\right. 
\end{equation*}%
In the following, suppose that $\theta _{1}(p_{1},p_{2})=\frac{p_{2}-p_{1}}{%
v_{2}-v_{1}}$ $\leq $ $\widetilde{\theta }$ so that the profit functions are
given by: 
\begin{eqnarray*}
\pi _{1} &=&\left( p_{1}-c_{1}\right) \left( \frac{p_{2}-p_{1}}{v_{2}-v_{1}}-%
\underline{\theta }\right) \frac{s}{(\widetilde{\theta }-\underline{\theta })%
},\text{ and} \\
\pi _{2} &=&\left( p_{2}-c_{2}\right) \left( \widetilde{\theta }-\frac{%
p_{2}-p_{1}}{v_{2}-v_{1}}\right) \frac{s}{(\widetilde{\theta }-\underline{%
\theta })}+\left( p_{2}-c_{2}\right) \left( \overline{\theta }-\widetilde{%
\theta }\right) \frac{1-s}{(\overline{\theta }-\widetilde{\theta })} \\
&=&\left( p_{2}-c_{2}\right) \left( \frac{\widetilde{\theta }\left(
v_{2}-v_{1}\right) s-sp_{2}+sp_{1}+\left( v_{2}-v_{1}\right) (\widetilde{%
\theta }-\underline{\theta })\left( 1-s\right) }{\left( v_{2}-v_{1}\right) (%
\widetilde{\theta }-\underline{\theta })}\right) .
\end{eqnarray*}%
The first-order conditions give the best response functions:%
\begin{eqnarray*}
\widehat{p}_{1} &=&\frac{1}{2}\left( p_{2}-\left( v_{2}-v_{1}\right) 
\underline{\theta }+c_{1}\right)  \\
\widehat{p}_{2} &=&\frac{\left( v_{2}-v_{1}\right) \left( \widetilde{\theta }%
-\underline{\theta }(1-s)\right) +sp_{1}+sc_{2}}{2s}
\end{eqnarray*}%
Combining gives the Nash prices and corresponding profits:%
\begin{eqnarray*}
p_{1}^{\ast } &=&\frac{\left( v_{2}-v_{1}\right) \left( \widetilde{\theta }-%
\underline{\theta }\left( 1+s\right) \right) +2sc_{1}+sc_{2}}{3s} \\
p_{2}^{\ast } &=&\frac{\left( v_{2}-v_{1}\right) \left( 2\left( \widetilde{%
\theta }-\underline{\theta }\right) +\underline{\theta }s\right)
+2sc_{2}+sc_{1}}{3s}
\end{eqnarray*}%
\begin{eqnarray*}
\pi _{1}^{\ast } &=&\left( \frac{\left( v_{2}-v_{1}\right) \left( \widetilde{%
\theta }-\underline{\theta }\left( 1+s\right) \right) -sc_{1}+sc_{2}}{3s}%
\right) \left( \frac{\left( v_{2}-v_{1}\right) \left( \widetilde{\theta }-%
\underline{\theta }\left( 1+s\right) \right) -sc_{1}+sc_{2}}{3\left(
v_{2}-v_{1}\right) (\widetilde{\theta }-\underline{\theta })}\right)  \\
&=&\frac{\left( p_{1}^{\ast }-c_{1}\right) ^{2}}{\left( v_{2}-v_{1}\right) }%
\frac{s}{(\widetilde{\theta }-\underline{\theta })}.
\end{eqnarray*}%
\begin{eqnarray*}
\pi _{2}^{\ast } &=&\left( \frac{\left( v_{2}-v_{1}\right) \left( 2\left( 
\widetilde{\theta }-\underline{\theta }\right) +\underline{\theta }s\right)
-sc_{2}+sc_{1}}{3s}\right) \left( \frac{\left( v_{2}-v_{1}\right) \left[
2\left( \widetilde{\theta }-\underline{\theta }\right) +\underline{\theta }s%
\right] -sc_{2}+sc_{1}}{3\left( v_{2}-v_{1}\right) (\widetilde{\theta }-%
\underline{\theta })}\right)  \\
&=&\frac{\left( p_{2}^{\ast }-c_{2}\right) ^{2}}{\left( v_{2}-v_{1}\right) }%
\frac{s}{(\widetilde{\theta }-\underline{\theta })}.
\end{eqnarray*}%
Under collusion, prices and profits are respectively given by:%
\begin{eqnarray*}
&&p_{1}^{c} \\
p_{2}^{c} &=&p_{2}^{\ast }+\left( p_{1}^{c}-p_{1}^{\ast }\right) =\frac{%
\left( v_{2}-v_{1}\right) \left( \widetilde{\theta }-\underline{\theta }+2%
\underline{\theta }s\right) +sc_{2}-sc_{1}+3sp_{1}^{c}}{3s}
\end{eqnarray*}%
and%
\begin{eqnarray*}
\pi _{1}^{c} &=&\left( p_{1}^{c}-c_{1}\right) \left( \frac{%
p_{2}^{c}-p_{1}^{c}}{v_{2}-v_{1}}-\underline{\theta }\right) \frac{s}{(%
\widetilde{\theta }-\underline{\theta })}=\left( p_{1}^{c}-c_{1}\right)
\left( \frac{p_{2}^{\ast }-p_{1}^{\ast }}{v_{2}-v_{1}}-\underline{\theta }%
\right) \frac{s}{(\widetilde{\theta }-\underline{\theta })} \\
&& \\
&=&\left( p_{1}^{c}-c_{1}\right) \left( \frac{\left( v_{2}-v_{1}\right)
\left( \widetilde{\theta }-\underline{\theta }\left( 1+s\right) \right)
+sc_{2}-sc_{1}}{3\left( v_{2}-v_{1}\right) (\widetilde{\theta }-\underline{%
\theta })}\right) =\frac{\left( p_{1}^{c}-c_{1}\right) \left( p_{1}^{\ast
}-c_{1}\right) }{\left( v_{2}-v_{1}\right) }\frac{s}{(\widetilde{\theta }-%
\underline{\theta })}. \\
&&
\end{eqnarray*}%
\begin{eqnarray*}
\pi _{2}^{c} &=&\left( p_{2}^{c}-c_{2}\right) \left( \frac{\widetilde{\theta 
}\left( v_{2}-v_{1}\right) s-s\left( p_{2}^{c}-p_{1}^{c}\right) +\left(
v_{2}-v_{1}\right) (\widetilde{\theta }-\underline{\theta })\left(
1-s\right) }{\left( v_{2}-v_{1}\right) (\widetilde{\theta }-\underline{%
\theta })}\right)  \\
&=&\left( p_{2}^{c}-c_{2}\right) \left( \frac{\widetilde{\theta }\left(
v_{2}-v_{1}\right) s-s\left( p_{2}^{\ast }-p_{1}^{\ast }\right) +\left(
v_{2}-v_{1}\right) (\widetilde{\theta }-\underline{\theta })\left(
1-s\right) }{\left( v_{2}-v_{1}\right) (\widetilde{\theta }-\underline{%
\theta })}\right)  \\
&=&\left( p_{2}^{c}-c_{2}\right) \left( \frac{3\widetilde{\theta }\left(
v_{2}-v_{1}\right) s-\left( v_{2}-v_{1}\right) \left( \widetilde{\theta }-%
\underline{\theta }+2\underline{\theta }s\right) -sc_{2}+sc_{1}+3\left(
v_{2}-v_{1}\right) (\widetilde{\theta }-\underline{\theta })\left(
1-s\right) }{3\left( v_{2}-v_{1}\right) (\widetilde{\theta }-\underline{%
\theta })}\right)  \\
&=&\left( p_{2}^{c}-c_{2}\right) \left( \frac{\left( v_{2}-v_{1}\right) %
\left[ 2\left( \widetilde{\theta }-\underline{\theta }\right) +\underline{%
\theta }s\right] -sc_{2}+sc_{1}}{3\left( v_{2}-v_{1}\right) (\widetilde{%
\theta }-\underline{\theta })}\right) =\frac{\left( p_{2}^{c}-c_{2}\right)
\left( p_{2}^{\ast }-c_{2}\right) }{\left( v_{2}-v_{1}\right) }\frac{s}{(%
\widetilde{\theta }-\underline{\theta })}.
\end{eqnarray*}%
Finally, deviating price and profits are:%
\begin{eqnarray*}
p_{1}^{d} &=&\frac{1}{2}\left( p_{2}^{c}-\left( v_{2}-v_{1}\right) 
\underline{\theta }+c_{1}\right)  \\
p_{2}^{d} &=&\frac{\left( v_{2}-v_{1}\right) \left( \widetilde{\theta }-%
\underline{\theta }(1-s)\right) +sp_{1}^{c}+sc_{2}}{2s}
\end{eqnarray*}%
and%
\begin{eqnarray*}
\pi _{1}^{d} &=&\left( p_{1}^{d}-c_{1}\right) \left( \frac{%
p_{2}^{c}-p_{1}^{d}}{v_{2}-v_{1}}-\underline{\theta }\right) \frac{s}{(%
\widetilde{\theta }-\underline{\theta })}= \\
&& \\
&=&\frac{1}{2}\left( p_{2}^{c}-\left( v_{2}-v_{1}\right) \underline{\theta }%
-c_{1}\right) \left( \frac{p_{2}^{c}-\left( v_{2}-v_{1}\right) \underline{%
\theta }-c_{1}}{2\left( v_{2}-v_{1}\right) }\right) \frac{s}{(\widetilde{%
\theta }-\underline{\theta })} \\
&& \\
&=&\frac{\left( p_{1}^{d}-c_{1}\right) ^{2}}{\left( v_{2}-v_{1}\right) }%
\frac{s}{(\widetilde{\theta }-\underline{\theta })}, \\
&&
\end{eqnarray*}%
and,%
\begin{eqnarray*}
\pi _{2}^{d} &=&\left( p_{2}^{d}-c_{2}\right) \left( \frac{\widetilde{\theta 
}\left( v_{2}-v_{1}\right) s-s\left( p_{2}^{d}-p_{1}^{c}\right) +\left(
v_{2}-v_{1}\right) (\widetilde{\theta }-\underline{\theta })\left(
1-s\right) }{\left( v_{2}-v_{1}\right) (\widetilde{\theta }-\underline{%
\theta })}\right)  \\
&& \\
&=&\left( \frac{\left( v_{2}-v_{1}\right) \left( \widetilde{\theta }-%
\underline{\theta }(1-s)\right) +sp_{1}^{c}-sc_{2}}{2s}\right) \left( \frac{%
\left( v_{2}-v_{1}\right) \left[ \widetilde{\theta }-\underline{\theta }%
\left( 1-s\right) \right] +sp_{1}^{c}-sc_{2}}{2\left( v_{2}-v_{1}\right) (%
\widetilde{\theta }-\underline{\theta })}\right)  \\
&& \\
&=&\frac{\left( p_{2}^{d}-c_{2}\right) ^{2}}{\left( v_{2}-v_{1}\right) }%
\frac{s}{(\widetilde{\theta }-\underline{\theta })}.
\end{eqnarray*}%
Combining to obtain the critical discount factor of each firm gives:%
\begin{eqnarray*}
\overline{\delta }_{1} &=&\dfrac{\pi _{1}^{d}-\pi _{1}^{c}}{\pi _{1}^{d}-\pi
_{1}^{\ast }}=\dfrac{\dfrac{\left( p_{1}^{d}-c_{1}\right) ^{2}}{\left(
v_{2}-v_{1}\right) }\dfrac{s}{(\widetilde{\theta }-\underline{\theta })}-%
\dfrac{\left( p_{1}^{c}-c_{1}\right) \left( p_{1}^{\ast }-c_{1}\right) }{%
\left( v_{2}-v_{1}\right) }\dfrac{s}{(\widetilde{\theta }-\underline{\theta }%
)}}{\dfrac{\left( p_{1}^{d}-c_{1}\right) ^{2}}{\left( v_{2}-v_{1}\right) }%
\dfrac{s}{(\widetilde{\theta }-\underline{\theta })}-\dfrac{\left(
p_{1}^{\ast }-c_{1}\right) ^{2}}{\left( v_{2}-v_{1}\right) }\dfrac{s}{(%
\widetilde{\theta }-\underline{\theta })}} \\
&& \\
&=&\frac{\left( p_{1}^{d}-c_{1}\right) ^{2}-\left( p_{1}^{c}-c_{1}\right)
\left( p_{1}^{\ast }-c_{1}\right) }{\left( p_{1}^{d}-c_{1}\right)
^{2}-\left( p_{1}^{\ast }-c_{1}\right) ^{2}}=\frac{\dfrac{1}{4}\left(
p_{1}^{c}-p_{1}^{\ast }\right) ^{2}}{\left( p_{1}^{c}-p_{1}^{\ast }\right)
\left( \frac{1}{4}\left( p_{1}^{c}-p_{1}^{\ast }\right) +\left( p_{1}^{\ast
}-c_{1}\right) \right) } \\
&& \\
&=&\frac{\dfrac{1}{4}\left( p_{1}^{c}-p_{1}^{\ast }\right) }{\dfrac{1}{4}%
\left( p_{1}^{c}-p_{1}^{\ast }\right) +\left( p_{1}^{\ast }-c_{1}\right) }.
\end{eqnarray*}

and%
\begin{eqnarray*}
\overline{\delta }_{2} &=&\dfrac{\pi _{2}^{d}-\pi _{2}^{c}}{\pi _{2}^{d}-\pi
_{2}^{\ast }}=\dfrac{\dfrac{\left( p_{2}^{d}-c_{2}\right) ^{2}}{\left(
v_{2}-v_{1}\right) }\dfrac{s}{(\widetilde{\theta }-\underline{\theta })}-%
\dfrac{\left( p_{2}^{c}-c_{2}\right) \left( p_{2}^{\ast }-c_{2}\right) }{%
\left( v_{2}-v_{1}\right) }\dfrac{s}{(\widetilde{\theta }-\underline{\theta }%
)}}{\dfrac{\left( p_{2}^{d}-c_{2}\right) ^{2}}{\left( v_{2}-v_{1}\right) }%
\dfrac{s}{(\widetilde{\theta }-\underline{\theta })}-\dfrac{\left(
p_{2}^{\ast }-c_{2}\right) ^{2}}{\left( v_{2}-v_{1}\right) }\dfrac{s}{(%
\widetilde{\theta }-\underline{\theta })}} \\
&& \\
&=&\dfrac{\left( p_{2}^{d}-c_{2}\right) ^{2}-\left( p_{2}^{c}-c_{2}\right)
\left( p_{2}^{\ast }-c_{2}\right) }{\left( p_{2}^{d}-c_{2}\right)
^{2}-\left( p_{2}^{\ast }-c_{2}\right) ^{2}}=\frac{\dfrac{1}{4}\left(
p_{1}^{c}-p_{1}^{\ast }\right) ^{2}}{\left( p_{1}^{c}-p_{1}^{\ast }\right)
\left( \frac{1}{4}\left( p_{1}^{c}-p_{1}^{\ast }\right) +\left( p_{2}^{\ast
}-c_{2}\right) \right) } \\
&& \\
&=&\frac{\dfrac{1}{4}\left( p_{1}^{c}-p_{1}^{\ast }\right) }{\dfrac{1}{4}%
\left( p_{1}^{c}-p_{1}^{\ast }\right) +\left( p_{2}^{\ast }-c_{2}\right) }.
\end{eqnarray*}

\section{Appendix 3:\ Comparisons with Other Models}

Symeonidis (1999) analyzes a representative consumer model, which differs
from ours in many ways. Meaningful comparisons are therefore difficult to
make. H\"{a}ckner's (1994) model is also different, but in many ways
comparable. Specifically, he is using a utility specification of the form $%
U\left( \theta \right) =v_{i}\left( \theta -p_{i}\right) $, whereas we
follow Mussa and Rosen (1978) which uses $U\left( \theta \right)
=v_{i}\theta -p_{i}$. To address this issue, we have tried to clarify the
differences between the different settings in the introduction. Moreover, we
have performed the same analysis in H\"{a}ckner's (1994) setting. This
analysis is presented below, but let us first summarize the main conclusion.
The main finding of Proposition 1 may not generally hold in H\"{a}ckner's
(1994) model. In particular, conclusions may be different when the highest
quality firms have the lowest profit margins and \textit{vice versa}. Yet,
the results of Proposition 1 do hold when the noncollusive profit margin is
increasing in quality. Moreover, if differences in unit costs are
sufficiently small, then it is indeed the lowest-quality seller who is most
inclined to deviate, all else unchanged. Consequently, the result of the
Corollary also holds when taking H\"{a}ckner's (1994) approach. We have
added a footnote highlighting this point (footnote 8). Consider an
intermediate firm in an $n$-firm variant of H\"{a}ckner's (1994) model with
cost heterogeneity. A consumer located at $\theta _{i}$ is indifferent
between buying from firm $i$ and $i+1$ when:%
\begin{eqnarray*}
\theta _{i}v_{i+1}-v_{i+1}p_{i+1} &=&\theta _{i}v_{i}-v_{i}p_{i}\text{ or} \\
\theta _{i} &=&\frac{v_{i+1}p_{i+1}-v_{i}p_{i}}{v_{i+1}-v_{i}}.
\end{eqnarray*}%
The profit function is then given by%
\begin{equation*}
\pi _{i}=\left( p_{i}-c_{i}\right) \left( \frac{v_{i+1}\left(
v_{i}-v_{i-1}\right) p_{i+1}+v_{i-1}\left( v_{i+1}-v_{i}\right)
p_{i-1}-v_{i}\left( v_{i+1}-v_{i-1}\right) p_{i}}{\left(
v_{i+1}-v_{i}\right) \left( v_{i}-v_{i-1}\right) }\right) ,
\end{equation*}%
which yields the following best-response function:%
\begin{equation*}
\widehat{p}_{i}\left( p_{i-1},p_{i+1}\right) =\frac{v_{i+1}\left(
v_{i}-v_{i-1}\right) p_{i+1}+v_{i-1}\left( v_{i+1}-v_{i}\right)
p_{i-1}+v_{i}\left( v_{i+1}-v_{i-1}\right) c_{i}}{2v_{i}\left(
v_{i+1}-v_{i-1}\right) }.
\end{equation*}%
Following the same steps in the proof of Proposition 1, this gives demand:%
\begin{equation*}
\theta _{i}^{\ast }-\theta _{i-1}^{\ast }=\frac{v_{i}\left(
v_{i+1}-v_{i-1}\right) }{\left( v_{i+1}-v_{i}\right) \left(
v_{i}-v_{i-1}\right) }\left( p_{i}^{\ast }-c_{i}\right) 
\end{equation*}%
and Nash profits%
\begin{equation*}
\pi _{i}^{\ast }=\left( p_{i}^{\ast }-c_{i}\right) ^{2}\left( \frac{%
v_{i}\left( v_{i+1}-v_{i-1}\right) }{\left( v_{i+1}-v_{i}\right) \left(
v_{i}-v_{i-1}\right) }\right) .
\end{equation*}%
In a similar fashion, it can be shown that:%
\begin{equation*}
\pi _{i}^{d}=\left( p_{i}^{d}-c_{i}\right) ^{2}\left( \frac{v_{i}\left(
v_{i+1}-v_{i-1}\right) }{\left( v_{i+1}-v_{i}\right) \left(
v_{i}-v_{i-1}\right) }\right) .
\end{equation*}%
As to collusive profits, we know that (due to the fixed market share rule):%
\begin{eqnarray*}
\theta _{i}^{\ast }v_{i+1}-v_{i+1}p_{i+1}^{c} &=&\theta _{i}^{\ast
}v_{i}-v_{i}p_{i}^{c}\text{ or} \\
p_{i+1}^{c} &=&p_{i+1}^{\ast }+\frac{v_{i}}{v_{i+1}}\left(
p_{i}^{c}-p_{i}^{\ast }\right) 
\end{eqnarray*}%
In general, the collusive price is then:%
\begin{equation*}
p_{i}^{c}=p_{i}^{\ast }+\frac{v_{1}}{v_{i}}\left( p_{1}^{c}-p_{1}^{\ast
}\right) .
\end{equation*}%
Following the covered market assumption, we know that $\theta
_{i}^{c}-\theta _{i-1}^{c}=\theta _{i}^{\ast }-\theta _{i-1}^{\ast }$ and
therefore,%
\begin{equation*}
\pi _{i}^{c}=\left( p_{i}^{c}-c_{i}\right) \left( \theta _{i}^{c}-\theta
_{i-1}^{c}\right) =\left( p_{i}^{\ast }+\frac{v_{1}}{v_{i}}\left(
p_{1}^{c}-p_{1}^{\ast }\right) -c_{i}\right) \left( \frac{%
v_{i}(v_{i+1}-v_{i-1})}{\left( v_{i+1}-v_{i}\right) \left(
v_{i}-v_{i-1}\right) }\left( p_{i}^{\ast }-c_{i}\right) \right) .
\end{equation*}%
Moreover,%
\begin{eqnarray*}
p_{i}^{d} &=&\frac{v_{i+1}\left( v_{i}-v_{i-1}\right)
p_{i+1}^{c}+v_{i-1}\left( v_{i+1}-v_{i}\right) p_{i-1}^{c}+v_{i}\left(
v_{i+1}-v_{i-1}\right) c_{i}}{2v_{i}\left( v_{i+1}-v_{i-1}\right) } \\
&=&\frac{v_{i+1}\left( v_{i}-v_{i-1}\right) \left( p_{i+1}^{\ast }+\frac{%
v_{1}\left( p_{1}^{c}-p_{1}^{\ast }\right) }{v_{i+1}}\right) +v_{i-1}\left(
v_{i+1}-v_{i}\right) \left( p_{i-1}^{\ast }+\frac{v_{1}\left(
p_{1}^{c}-p_{1}^{\ast }\right) }{v_{i-1}}\right) +v_{i}\left(
v_{i+1}-v_{i-1}\right) c_{i}}{2v_{i}\left( v_{i+1}-v_{i-1}\right) } \\
&=&\frac{v_{i+1}\left( v_{i}-v_{i-1}\right) p_{i+1}^{\ast }+v_{i-1}\left(
v_{i+1}-v_{i}\right) p_{i-1}^{\ast }+\left( v_{i+1}-v_{i-1}\right)
v_{1}\left( p_{1}^{c}-p_{1}^{\ast }\right) +v_{i}\left(
v_{i+1}-v_{i-1}\right) c_{i}}{2v_{i}\left( v_{i+1}-v_{i-1}\right) } \\
&=&p_{i}^{\ast }+\frac{1}{2}\frac{v_{1}}{v_{i}}\left( p_{1}^{c}-p_{1}^{\ast
}\right) .
\end{eqnarray*}%
Combining the above profit specifications gives%
\begin{eqnarray*}
\pi _{i}^{d}-\pi _{i}^{c} &=&\frac{v_{i}\left( v_{i+1}-v_{i-1}\right) \left(
p_{i}^{d}-c_{i}\right) ^{2}}{\left( v_{i+1}-v_{i}\right) \left(
v_{i}-v_{i-1}\right) }-\frac{v_{i}(v_{i+1}-v_{i-1})\left( p_{i}^{\ast
}-c_{i}\right) \left( p_{i}^{\ast }+\frac{v_{1}}{v_{i}}\left(
p_{1}^{c}-p_{1}^{\ast }\right) -c_{i}\right) }{\left( v_{i+1}-v_{i}\right)
\left( v_{i}-v_{i-1}\right) } \\
&=&\frac{v_{i}\left( v_{i+1}-v_{i-1}\right) \left[ \left( p_{i}^{\ast }+%
\frac{1}{2}\frac{v_{1}}{v_{i}}\left( p_{1}^{c}-p_{1}^{\ast }\right)
-c_{i}\right) ^{2}-\left( p_{i}^{\ast }-c_{i}\right) ^{2}-\frac{v_{1}}{v_{i}}%
\left( p_{1}^{c}-p_{1}^{\ast }\right) \left( p_{i}^{\ast }-c_{i}\right) %
\right] }{\left( v_{i+1}-v_{i}\right) \left( v_{i}-v_{i-1}\right) } \\
&=&\frac{v_{i}\left( v_{i+1}-v_{i-1}\right) \left[ \frac{1}{4}\left( \frac{%
v_{1}}{v_{i}}\right) ^{2}\left( p_{1}^{c}-p_{1}^{\ast }\right) ^{2}\right] }{%
\left( v_{i+1}-v_{i}\right) \left( v_{i}-v_{i-1}\right) }.
\end{eqnarray*}%
and%
\begin{eqnarray*}
\pi _{i}^{d}-\pi _{i}^{\ast } &=&\frac{v_{i}\left( v_{i+1}-v_{i-1}\right)
\left( p_{i}^{d}-c_{i}\right) ^{2}}{\left( v_{i+1}-v_{i}\right) \left(
v_{i}-v_{i-1}\right) }-\frac{v_{i}(v_{i+1}-v_{i-1})\left( p_{i}^{\ast
}-c_{i}\right) ^{2}}{\left( v_{i+1}-v_{i}\right) \cdot \left(
v_{i}-v_{i-1}\right) } \\
&=&\frac{v_{i}(v_{i+1}-v_{i-1})\left( \left( p_{i}^{\ast }+\frac{1}{2}\frac{%
v_{1}}{v_{i}}\left( p_{1}^{c}-p_{1}^{\ast }\right) -c_{i}\right) ^{2}-\left(
p_{i}^{\ast }-c_{i}\right) ^{2}\right) }{\left( v_{i+1}-v_{i}\right) \cdot
\left( v_{i}-v_{i-1}\right) } \\
&=&\frac{v_{i}(v_{i+1}-v_{i-1})\left( \frac{1}{4}\left( \frac{v_{1}}{v_{i}}%
\right) ^{2}\left( p_{1}^{c}-p_{1}^{\ast }\right) ^{2}+\frac{v_{1}}{v_{i}}%
\left( p_{1}^{c}-p_{1}^{\ast }\right) \left( p_{i}^{\ast }-c_{i}\right)
\right) }{\left( v_{i+1}-v_{i}\right) \cdot \left( v_{i}-v_{i-1}\right) } \\
&=&\frac{v_{i}(v_{i+1}-v_{i-1})\left( \frac{v_{1}}{v_{i}}\left(
p_{1}^{c}-p_{1}^{\ast }\right) \right) \left( \frac{1}{4}\left( \frac{v_{1}}{%
v_{i}}\right) \left( p_{1}^{c}-p_{1}^{\ast }\right) +\left( p_{i}^{\ast
}-c_{i}\right) \right) }{\left( v_{i+1}-v_{i}\right) \cdot \left(
v_{i}-v_{i-1}\right) }.
\end{eqnarray*}%
Thus, the critical discount factor of an intermediate firm $i$ in H\"{a}%
ckner's (1994) setting with cost heterogeneity and $n$ firms is: 
\begin{eqnarray*}
\delta  &\geq &\overline{\delta }_{i}=\frac{\pi _{i}^{d}-\pi _{i}^{c}}{\pi
_{i}^{d}-\pi _{i}^{\ast }}=\frac{\left( \frac{v_{i}\left(
v_{i+1}-v_{i-1}\right) }{\left( v_{i+1}-v_{i}\right) \left(
v_{i}-v_{i-1}\right) }\right) \left[ \frac{1}{4}\left( \frac{v_{1}}{v_{i}}%
\right) ^{2}\left( p_{1}^{c}-p_{1}^{\ast }\right) ^{2}\right] }{\left( \frac{%
v_{i}(v_{i+1}-v_{i-1})}{\left( v_{i+1}-v_{i}\right) \cdot \left(
v_{i}-v_{i-1}\right) }\right) \left( \frac{v_{1}}{v_{i}}\left(
p_{1}^{c}-p_{1}^{\ast }\right) \right) \left( \frac{1}{4}\left( \frac{v_{1}}{%
v_{i}}\right) \left( p_{1}^{c}-p_{1}^{\ast }\right) +\left( p_{i}^{\ast
}-c_{i}\right) \right) } \\
&=&\frac{\frac{1}{4}\left( \frac{v_{1}}{v_{i}}\right) \left(
p_{1}^{c}-p_{1}^{\ast }\right) }{\frac{1}{4}\left( \frac{v_{1}}{v_{i}}%
\right) \left( p_{1}^{c}-p_{1}^{\ast }\right) +\left( p_{i}^{\ast
}-c_{i}\right) }=\frac{\frac{1}{4}v_{1}\left( p_{1}^{c}-p_{1}^{\ast }\right) 
}{\frac{1}{4}v_{1}\left( p_{1}^{c}-p_{1}^{\ast }\right) +v_{i}\left(
p_{i}^{\ast }-c_{i}\right) }.
\end{eqnarray*}%
Results for the highest and lowest quality firm can be derived in a similar
way (see the proof of Proposition 1).

Recall that the corresponding critical discount factor in Proposition 1 is
given by%
\begin{equation*}
\delta \geq \overline{\delta }_{i}=\frac{\pi _{i}^{d}-\pi _{i}^{c}}{\pi
_{i}^{d}-\pi _{i}^{\ast }}=\frac{\frac{1}{4}\left( p_{1}^{c}-p_{1}^{\ast
}\right) }{\frac{1}{4}\left( p_{1}^{c}-p_{1}^{\ast }\right) +\left(
p_{i}^{\ast }-c_{i}\right) }.
\end{equation*}%
Observe that, compared to the result of Proposition 1, it is now possible
that a firm with a higher competitive profit margin is more eager to leave
the cartel when it is of sufficiently low quality (\textit{i.e.}, a high $%
p_{i}^{\ast }-c_{i}$ may be more than neutralized by a low $v_{i}$) and 
\textit{vice versa}. The result of Proposition 1 does apply, however, when
the noncollusive price-cost margin is increasing in quality. Moreover, and
in contrast to H\"{a}ckner, it is indeed the lowest-quality supplier who has
the tightest incentive constraint when differences in unit costs are
sufficiently small (in accordance with the result of the Corollary). As to
the latter, note that H\"{a}ckner (1994)\ considers a duopoly without costs
and does not impose a fixed market share rule, which is an important driver
of our finding.


\begin{thebibliography}{99}
\bibitem{} Ecchia, Giulio and Luca Lambertini (1997), \textquotedblleft
Minimum Quality Standards and Collusion,\textquotedblright\ \textit{Journal
of Industrial Economics}, 45(1), 101-113.

\bibitem{} Friedman, James W. (1991), \textit{Game Theory with\ Applications
to Economics}. Oxford University Press, Oxford.

\bibitem{} Gabszewicz, Jean J. and Jacques-Fran\c{c}ois Thisse (1979),
\textquotedblleft Price Competition, Quality and Income
Disparities,\textquotedblright\ \textit{Journal of Economic Theory}, 20(3),
340-359.

\bibitem{} H\"{a}ckner, Jonas (1994), \textquotedblleft Collusive pricing in
markets for vertically differentiated products,\textquotedblright\ \textit{%
International Journal of Industrial Organization}, 12(2), 155-177.

\bibitem{} Harrington, Joseph E. Jr. (2006), \ \textquotedblleft How do
Cartels Operate?,\textquotedblright\ \textit{Foundations and Trends in
Microeconomics}, 2(1), 1-105.

\bibitem{} Marini, Marco A. (2018), "Collusive Agreements in Vertically
Differentiated Markets", in \textit{Handbook of Game Theory and Industrial
Organization}, \textit{Volume 2:\ Applications}, L. C. Corchon and M. A.
Marini (eds.), Edward Elgar, Chelthenam, UK, Northampton, MA, USA.

\bibitem{} Mussa, Michael and Sherwin Rosen (1978), \textquotedblleft
Monopoly and Product Quality,\textquotedblright\ \textit{Journal of Economic
Theory}, 18(2), 301-317.

\bibitem{} Shaked, Avner and John Sutton (1982), \textquotedblleft Relaxing
Price Competition Through Product Differentiation,\textquotedblright\ 
\textit{Review of Economic Studies}, 49 (1), 3-13.

\bibitem{} Symeonidis, George (1999), \textquotedblleft Cartel stability in
advertising-intensive and R\&D-intensive industries,\textquotedblright\ 
\textit{Economics Letters}, 62, 121-129.

\bibitem{} Tirole, Jean (1988), \textquotedblleft \textit{The Theory of
Industrial Organization},\textquotedblright\ MIT Press, Cambridge,
Massachusetts.

\bibitem{} Vives, Xavier (2000), \textit{Oligopoly Pricing}. \textit{Old
Ideas and New Tools}. MIT\ Press, Cambridge, Massachusetts.
\end{thebibliography}
\end{document}